\documentclass[
 reprint,
 amsmath,
 amssymb,
 aps,
]{revtex4-1}           
	\usepackage{array}
	\usepackage{amsthm}
	\usepackage{xcolor}
	\usepackage{tikz}
	\usepackage{amssymb}
	\usepackage{graphicx}
	\usepackage{hyperref}
	\hypersetup{
     colorlinks=true,
     linkcolor=blue,
     filecolor=blue,
     citecolor=magenta,      
     urlcolor=cyan,
     }
	\usepackage{mathtools}
	\usepackage{subcaption}
	\newtheorem{theorem}{Theorem}
	\newtheorem{definition}{Definition}
	\newtheorem{proposition}{Proposition}
\begin{document}
%	------------------------------------------------------------/
%
%
%	this is a blank space
%
%
%	------------------------------------------------------------/
\title{On how generalised entropies without parameters impact information optimisation processes}
\author{Jes\'us Fuentes}\email{j.fuentesaguilar@ugto.mx}
\author{Octavio Obreg\'on} 
\affiliation{Departamento de F\'{\i}sica, Divisi\'on de Ciencias e Ingenier\'ias Campus Le\'on,
  \\ Universidad de Guanajuato, A.P. E-143, C.P. 37150, Le\'on, Guanajuato, M\'exico.}
%\date{}
%	------------------------------------------------------------/
%
%
%	this is a blank space
%
%
%	------------------------------------------------------------/
\begin{abstract}

As an application of generalised statistical mechanics, it is studied a possible route toward a consistent generalised information theory in terms of a family of non-extensive, non-parametric entropies $H^\pm_D(P)$. Unlike other proposals based on non-extensive entropies with a parameter dependence, our scheme is  asymptotically equivalent to the one formulated by Shannon, while it differs in regions where the density of states is reasonably small, which leads to information distributions constrained to their background. Two basic concepts are discussed to this aim. First, we prove two effective coding theorems for the entropies $H^\pm_D(P)$. Then we calculate the channel capacity of a binary symmetric channel (BSC) and a binary erasure channel (BEC) in terms of these entropies. We found that processes such as data compression and channel capacity maximisation can be improved in regions where there is a low density of states, whereas for high densities our results coincide with the Shannon's formulation.

\begin{description}
\item[PACS numbers]
02.50.-r; 89.70.+c; 05.40.Ca
\item[keywords] data compression, coding theorems, non-extensive entropies
\end{description}
\end{abstract}
%	------------------------------------------------------------/
%
%
%	this is a blank space
%
%
%	------------------------------------------------------------/
\maketitle
%	------------------------------------------------------------/
%
%
%	this is a blank space
%
%
%	------------------------------------------------------------/
\section{Introduction}

The first rigorous study on data compression was published by Claude Shannon in 1948 \cite{shannon}, in a seminal paper that constituted a firm setting in the foundations of classical information theory. In that work it was considered a $D$-ary alphabet as well as a collection of uncorrelated random sources of letters (whose probabilities depend only on the $k$ letters that immediately precede them) with the aim to examine to what extent could one compress $N$ symbols emitted from a particular source. Shannon found that any process of data compression results limited by the entropy itself such that a codeword whose average length attains the entropy can be regarded as optimal. In other words, this result led to quantify the average amount of information (absence of redundancy) by simply measuring the entropy.

Nevertheless, the entropy as an information measure may have different definitions in its functional form and conceptual purport. Although, in general, we can classify them into two different groups: extensive and non-extensive entropies. The class of non-extensive entropies is typically related to statistical systems out of equilibrium and  whose components are interacting with themselves or with external agents. These systems are often referred as non-equilibrium processes, and play a direct role in modern trends of information theory \cite{crisanti}. In this sense, a non-extensive entropy can be understood as a generalisation to the (Boltzmann-Gibbs) Shannon's entropy
\begin{equation*}
H_D^S(P)=-\sum_ip(x_i)\ln(p(x_i)),
\end{equation*}
since the latter succeeds as long as a large system can be treated in equilibrium. Which shall not be disregarded given the equilibrium configuration is realisable for a large number of situations up to a good approximation.  

Thus the question whether the introduction of non-extensive entropies in those systems that were hardly approximated in equilibrium can really bring some significant novelty, has a positive answer and, in fact, is worth attention even when the description of the problem turns complicated via a non-extensive measure. 

In information theory the story is not different. Below we shall show how the rate of data compression for distributions associated with small systems (i.e. setups with a few number of accessible micro states $N$) can attain a more efficient bound compared with the standard formulation if the process is accomplished by means of the non-extensive entropy $H_D^-(P)$ application, whereas $H_D^-(P)$ and $H_D^S(P)$ will tend to coincide as the number of the accessible states $N$ grows. Actually, the entropies $H_D^\pm(P)$ and $H_D^S(P)$ are asymptotically equivalent as seen from Fig. \ref{f:plot}, see Ref. \cite{paper} for further details. This means that the entropies $H_D^\pm(P)$ truly resemble the standard theory whenever the probability distribution $P$ concerns a large system. 

A similar discussion will be addressed to the computation of the channel capacity of a BSC and a BEC.  In the first case, we will show how the entropy $H_D^+(P)$ improves the outcome obtained from $H_D^S(P)$, recalling that the partial optimisation will be valid in the regime of small systems. In the case of a BEC, there is an interesting switching on the obtained bounds, going from upper to lower and vice versa, depending on what ratio of errors is present over the communication channel. 

In turn, we shall introduce the functional form of the entropies $H_D^\pm(P)$:
\begin{equation}
\label{entropies}
H_D^{\pm}(P)=-\sum_i^{N}p(x_i)\log^{\pm}_{D}(p(x_i)),
\end{equation}
where $\log^{\pm}$ are generalised logarithms (see Appendix \ref{app}) and $\{p(x_1),\ldots,p(x_N)\}\in P$ are the probabilities of the codewords  $x_1,\ldots,x_N$ emitted by an uncorrelated random source $X$. The subindex $D$ refers to the size of the alphabet, e.g. a binary alphabet has $D=2$ symbols. 

The information measures \eqref{entropies} were originally proposed by one of us in \cite{oo10}, within the superstatistics framework \cite{beck} and years after, their  quantum equivalents were obtained by means of a generalised replica trick \cite{oo18}, i.e. a useful method to compute the partition function in spin crystals.

Indeed, the entropies $H_D^\pm(P)$ demand another generalised quantities as we will see below. For instance, we cannot find effective coding theorems in terms of these measures if the constraint $\sum_i^N D^{-l_i}\leq1$ is imposed to minimise the corresponding codeword length, which leads to the necessity of a generalised constraint, as stated in Prop. \ref{p1}.

On the basis of the entropies \eqref{entropies}, a third one can be observed as well, expressly $H_D^{0}(P)=\frac{1}{2}\left[H_D^{+}(P)+H_D^{-}(P)\right]$. Yet along this paper we are not to analyse explicitly this latter measure, but only the fundamental entropies $H_D^\pm(P)$, nonetheless, given that $H_D^{0}(P)$ is an average of the other two entropy measures, the following discussion directly comprises this case. 

Note that $H_D^\pm(P)$ do not depend on any free parameter but only on the probability distribution. Both non-extensive entropies, furthermore, resemble the Shannon's entropy at first order, which is evident from the following series expansion:
\begin{equation}
\begin{split}
-p&(x_i)\log^{\pm}_{D}(p(x_i)) = -p(x_i)\log_D(p(x_i)) \\ 
&\mp \frac{1}{2}[p(x_i)\log_D(p(x_i))]^2 - \frac{1}{6}[p(x_i)\log_D(p(x_i))]^3 + \cdots 
\end{split}
\end{equation}
due to $p(x_i) \in [0,1]$ the higher order terms become subdominant but their actual contribution is not entirely negligible in a regime of high probabilities.

The special case of a uniform distribution $p(x_i)=\frac{1}{N}$ gives us further information about the behaviour of $H_D^\pm(P)$. Substituting this distribution into \eqref{entropies} yields $H_D^{\pm}(P)=\pm N \mp N^{1\mp\frac{1}{N}}$, and its series expansion:
\begin{equation}
\label{equi}
H_D^{\pm}(P) = \sum_i^\infty(\mp1)^{i+1}\frac{\log^iN}{i!N^{i-1}},
\end{equation}
which produces the plots in Fig. \ref{f:plot}. Our previous discussion can be summarised in that figure. 

From an intuitive viewpoint, one could ask to what degree could the standard information theory be modified by implementing the information measures \eqref{entropies}, since Fig. \ref{f:plot} tells in advance the existence of upper and lower bounds on the Shannon's entropy in certain regions of $N$.

\begin{figure}[!h]	
	\centering
	\includegraphics[width=0.5\textwidth]{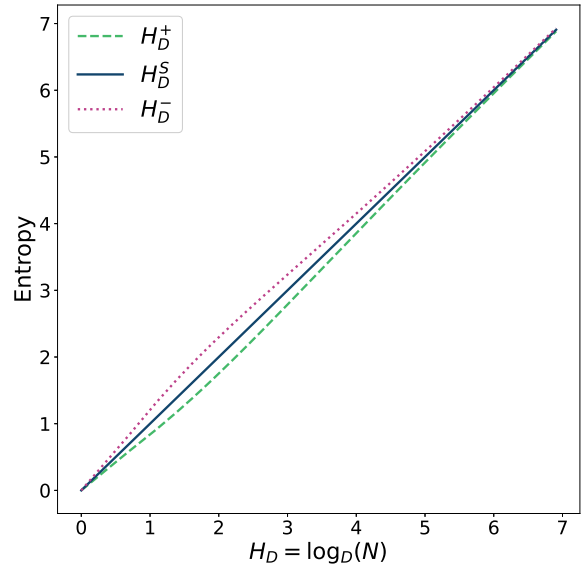}
	\caption{Behaviour of the entropies \eqref{entropies} with respect to the Shannon's entropy in terms of a uniform distribution $p(x_i)=\frac{1}{N}$ for all $x_i\in X$.}
	\label{f:plot}
\end{figure}

As well, it is worthy of note that the Shannon's entropy and the entropies $H_D^\pm(P)$ are written in the fashion of  $H_D(P)=\sum_i^NG_D(p(x_i))$, where $G_D(x)$ is a strictly monotonic increasing function known as the entropic form. There are other entropy measures that adopts such structure, such as the Tsallis's $q$-entropy \cite{tsallis} (see also \cite{havrda}), which has been widely studied in different disciplines \cite{physica}. Regarding information theory, a coding theorem in terms of the $q$-entropy has been found by means of an algebraic method \cite{chapeau}. Yet, it has been argued in \cite{abe10,oikonomou}, that its free parameter $q$ may be an obstacle in the theoretical explanation of some problems, as well as its divergence from the standard  entropy for a large number of symbols (or micro states) even if $q\approx1$. The hassle of free parameters is absent from $H_D^\pm(P)$, also it has been proved in \cite{paper} that $H_D^\pm(P)$ are stable measures for any $N$ and any $P$, thus we think such entropies are good candidates for a generalised information theory.
 
The structure of our work is as follows. In Sec. \ref{s:theorems} we obtain coding theorems for the two entropies defined in \eqref{entropies} and discuss the implications of the resulted data compression in comparison with the Shannon's entropy. Sec. \ref{s:capacity} is devoted to the computation of the channel capacity of a BSC and a BEC, we compare  our results with the ones in the standard theory. Finally, our conclusions and general outlook are presented in Sec. \ref{s:conclusions}.

%	------------------------------------------------------------/
%
%
%	this is a blank space
%
%
%	------------------------------------------------------------/
\section{Effective coding theorems}
\label{s:theorems}

Below we shall present coding theorems in terms of the non-extensive entropies \eqref{entropies}. To this aim, a variational method is applied to find the average  length that a codeword can attain by compressing the data subject to an adequate constraint. This constraint is not arbitrary. As we will show there is an intertwining between the entropy measure and the minimisation constraint such that the functional structure of the entropy is preserved. 

For our purposes, in the following discussion we assume that the communication channel is perfectly noiseless, that is we shall neglect any external perturbation. As well, we take for granted that the codes are prefix-free, i.e. there is no a single codeword which is the initial segment of any other available codeword in the set generated by $X$.

The aim of the method reported below is to find codes that minimise a monotonic increasing function of the form
\begin{equation}
\label{long}
L = \varphi^{-1}\left(\sum_i^{N} p(x_i) \varphi(l(x_i)) \right),
\end{equation}
where $l(x_i)\in\mathbb{Z}_+$ is the length of the $i$-th codeword related to the Nagumo-Kolmogorov function $\varphi$ such that the cost of managing a sequence of length $l(x_i)$ is $\varphi(l(x_i)):\mathbb{R}_+\to\mathbb{R}_+$. The quantity defined in \eqref{long} is called the {\it average  length for the cost function} $\varphi$, in accordance with Campbell \cite{campbell0}, however, for brevity we refer to it simply as the average length.

A number of cost functions $\varphi$ have been studied in \cite{baer}, although for most entropies of the generalised form, we have that $\varphi(z)=z$. For instance, in the case of the Shannon's entropy, the average length $L_S=\varphi^{-1}(\sum_i^{N}p_S(x_i)\varphi(l_S(x_i)))=\sum_i^{N}p_S(x_i)l_S(x_i)$ can be optimised via the functional
\begin{equation}
\label{class}
J_S=L_S +\gamma K_S,
\end{equation}
where $\gamma\in\mathbb{R}$ is an undetermined Lagrange multiplier and the constraint 
\begin{equation}
\label{kraft}
K_S=\sum_i^{N} D^{-l_S(x_i)}\leq1,
\end{equation}
is the Kraft (McMillan) inequality \cite{cover}, accounting for the fundamental control function in the optimisation of this case. 

To solve the optimisation problem in \eqref{class}, one merely differentiate $J_S$ with respect to $l_S(x_i)$ and equates to zero, implying that the rule to find the extrema of \eqref{class} is the same for every $x_i\in X$. After the elementary calculation one finds:
\begin{equation}
l_S^*(x_i)=-\log_Dp_S(x_i),
\end{equation}
where the superscript * indicates that such average length is regarded as optimal, namely it solves the problem \eqref{class}.

Below we are to proceed in an analogous way but in terms of the entropy measures \eqref{entropies}. Nonetheless, a generalisation to \eqref{kraft} shall be devised to optimise the mean lengths related to $H_D^\pm(P)$. 

This is an interesting aspect of our scheme, for the quantity to be minimised will grant a different rate of data compression than the estimated by the classical theory.

Hence, in our case, the problem consists of choosing codes that become minima when they are subject to the quantities in the following definition.

\begin{proposition}\label{p1}
The entropies $H_D^\pm(P)$ define an optimisable functional of the class \eqref{class}  
\begin{equation}
\label{functional1}
J_{\pm}=L_{\pm}+\gamma K_\pm
\end{equation}
where $\gamma$ is a Lagrange multiplier and the average lengths $L_\pm$ are univocally determined by the constraints:
\begin{equation}
\label{constraint1}
K_\pm =\sum_i^{N}\sum_j^{\infty} a_{\pm}(j) \Gamma[j+1,-\log D^{-l_\pm(x_i)}]\leq \text{const.},
\end{equation}
the individual lengths $l_\pm(x_i)$ will eventually be related to the probability distributions $p_\pm(x_i)$, the real coefficients $a_{\pm}(j)$ are given in Appendix \ref{app} and 
\begin{equation*}
\Gamma(y,x) =\int_x^\infty dz\,z^{y-1}e^{-z}, 
\end{equation*}
is the incomplete gamma function.
\end{proposition}

\begin{proof}
Since \eqref{entropies} can be expressed in the generalised form, the corresponding Nagumo-Kolmogorov function is $\varphi(x)=x$. Consistently, the average lengths in \eqref{functional1} read as
\begin{equation}
\label{length1}
L_{\pm}=\sum_i^{N} p_\pm(x_i) l_\pm(x_i).
\end{equation}

To find the optimal individual lengths $l_\pm^{*}(x_i)$, we now differentiate \eqref{functional1} with respect to $l_\pm(x_i)$, that is:
\begin{equation}
\frac{\partial J_{\pm}}{\partial l_\pm(x_i)} = \frac{\partial L_{\pm}}{\partial l_\pm(x_i)} + \gamma\frac{\partial K_\pm}{\partial l_\pm(x_i)},
\end{equation}
yet we are looking for a global minimum, it follows that the equality must vanish term to term, therefore
\begin{equation}
\label{method}
\begin{split}
\frac{\partial L_{\pm}}{\partial l_\pm(x_i)} &= \sum_i^{N}p_{\pm}(x_i)\\
&=-\gamma\frac{\partial K_\pm}{\partial l_\pm(x_i)}\\
&=\gamma \log D \sum_i^{N} \exp^{\pm}\left[-l_\pm(x_i)\log D\right]\\
&= \sum_i^{N} \exp^{\pm}\left[-l_\pm(x_i) \log D\right],
\end{split}
\end{equation}
where the Lagrange multiplier has been selected as $\gamma = \frac{1}{\log D}$ and $\exp^{\pm}$ are stretched exponential functions (see Appendix \ref{app}). 

One can observe the equality in \eqref{method} is satisfied for all $x_i \in X$ iff $l_{\pm}^*(x_i)=-\log^{\pm}_{D}\left(p_{\pm}(x_i)\right)$. At this point it should be evident that we have obtained two optimal individual lengths as given by each  probability distribution $p_{\pm}(x_i)$, for that reason we have appended the label $\pm$ to $l_\pm^*(x_i)$, since the lengths and the distributions are univocally related. That completes the proof.
\end{proof}

It follows that the expected lengths $L_{\pm}$ are lower and upper bounded by the entropy measures $H_D^\pm(P)$, which is consistent with the standard formulation. Thus we are entitled to introduce the coding Theorems \ref{t:h1} and \ref{t:h2}.

However, before the formal introduction of such theorems, some comments deserve attention. First, note that the individual lengths $l_\pm$ and the probabilities $p_\pm$ are determined by each other, therefore the constraints \eqref{constraint1} are not arbitrary, but must be constructed such that entropy measures can be recovered at the optimal points.

We also remark that the standard theory is straightforwardly recoverable from our scheme. First note how the Kraft inequality \eqref{kraft} is obtained by truncating $K_\pm$ at first order, that is:
\begin{equation}
\begin{split}
K_\pm&=a_\pm(0)\sum_i^N D^{-l_\pm(x_i)} \\
&\quad+ a_\pm(1)\sum_i^N \Gamma[2,-\log[D^{-l_\pm(x_i)}] \\
&\qquad+ a_\pm(2)\sum_i^N \Gamma[3,-\log[D^{-l_\pm(x_i)}] + \cdots,
\end{split}
\end{equation}
with $a_\pm(0)=1$, see Appendix \ref{app}.

Equivalently the limit $N\to\infty$ means that the system will possess pretty low probabilities due to a high number of accessible states, in that case our proposal will be asymptotically equivalent to the Shannon's theory \cite{paper}, symbolically expressed as $p_S(x)=\lim_{N\to\infty}p_{\pm}(x)$.
 
 %Therefore, the standard functional \eqref{class} is recovered from \eqref{functional1} as
%\begin{equation}
%\label{functional0}
%\lim_{N\to\infty}J_{\pm}=J_S.
%\end{equation}

\begin{theorem}\label{t:h1} The expected lengths defined by Eq. \eqref{length1} for a $D$-ary alphabet regarding the entropies $H_D^{\pm}(P)$, satisfy $L_{\pm}\geq H_D^{\pm}(P)$, with equality iff $l_{\pm}^*(x_i)=-\log^{\pm}_{D}\left(p_\pm(x_i)\right)$ for every $x_i$ in $X$.
\end{theorem}

\begin{proof} It follows directly by writing the difference between the expected lengths and the entropies. Then one gets
\begin{equation}
\begin{split}
L_{\pm}-H_D^{\pm}(P) &= \sum_i^{N} p_\pm(x_i) l_\pm(x_i) \\
&\hspace{1.3cm}+ \sum_i^{N} p_\pm(x_i) \log^{\pm}_{D}\left(p_\pm(x_i)\right)\\
& = \sum_i^{N} p_\pm(x_i)\left[l_\pm(x_i) + \log^{\pm}_{D}\left(p_\pm(x_i)\right)\right]\\
& \geq 0,
\end{split}
\end{equation}
necessarily leading to $l_\pm(x_i)\geq -\log^{\pm}_{D}\left(p_\pm(x_i)\right)$, for the reason that every $l_\pm(x_i)$ is an integer. Then, it follows that the equality is attained iff the individual lengths $l_\pm(x_i)=l_\pm^*(x_i)$ are optimal. And the theorem is demonstrated.
\end{proof}

Likewise, the entropies $H_D^{\pm}$ amount to a lower bound on  the expected lengths $L_{\pm}$, yet as we are to show these lengths are within one {\em dit} of the lower bound as well.

\begin{theorem}\label{t:h2} For a $D$-ary alphabet and a source distribution $X$ let $l_\pm(x_i)$ be the optimal individual lengths that solve the optimisation problem \eqref{functional1}, where the associated mean lengths are defined by Eq. \eqref{length1}. Then
\begin{equation}
\label{upper}
H_D^{\pm}(P)\leq L_{\pm}< H_D^{\pm}(P) + 1.
\end{equation}
\end{theorem}

\begin{proof} As stated by Theorem \ref{t:h1}, the choice of codeword lengths $l_\pm^*(x_i)=-\log^{\pm}_{D}\left(p_\pm(x_i)\right)$ results in $L_{\pm}^*=H_D^{\pm}$. However, to assure that every $l_{\pm}(x_i)$ is an integer, then we shall take $l_{\pm}(x_i)=\left\lceil-\log^{\pm}_{D}\left(p_{\pm}(x_i)\right)\right\rceil$. In that way the individual lengths do really  satisfy
\begin{equation}
\label{ineq1}
-\log^{\pm}_{D}\left(p_\pm(x_i)\right) \leq l_\pm(x_i) < - \log^{\pm}_{D}\left(p_\pm(x_i)\right) + 1,
\end{equation}
multiplying each member by $p_\pm(x_i)$ and summing over all $x_i$, leads to
\begin{equation}
\begin{split}
-\sum_i^{N}p_\pm(x_i)\log^{\pm}_{D}\left(p_\pm(x_i)\right)  & \leq \sum_i^{N} p_\pm(x_i)l_\pm(x_i) \\
& < - \sum_i^{N} p_\pm(x_i)\log^{\pm}_{D}\left(p_\pm(x_i)\right) \\
&\hspace{1.5cm} +\sum_i^{N}p_\pm(x_i),
\end{split}
\end{equation}
then from Eqs. \eqref{entropies} and \eqref{length1} together with  the normalisation of the probability, we finally arrive at the expression
\begin{equation}
H_D^{\pm}(P)\leq L_{\pm}< H_D^{\pm}(P) + 1,
\end{equation}
and we have the theorem.
\end{proof} 

We would like to remark that the expected lengths $L_{\pm}$ satisfy $H_D^{\pm} \leq L_{\pm} < H_D^{\pm} + 1$. Nonetheless the optimal code in view of $H_D^{\pm}$ can only be better that the prescribed by $L_{\pm}$, thus one is addressed to the Theorem \ref{t:h2}.

\begin{figure}[!h]
	\centering
	\includegraphics[width=0.5\textwidth]{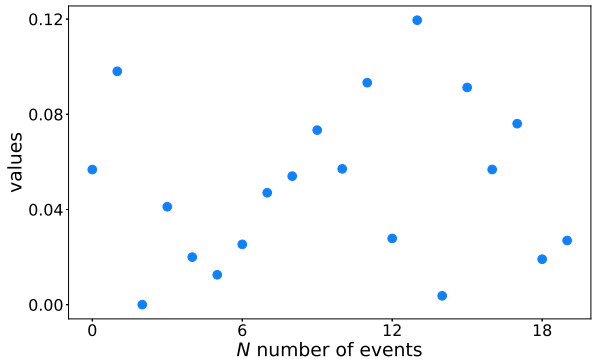}
	\caption{Random processes for $N = 20$ events.}
	\label{f:process1}
\end{figure}

\begin{figure}[!h]	
	\centering
	\includegraphics[width=0.5\textwidth]{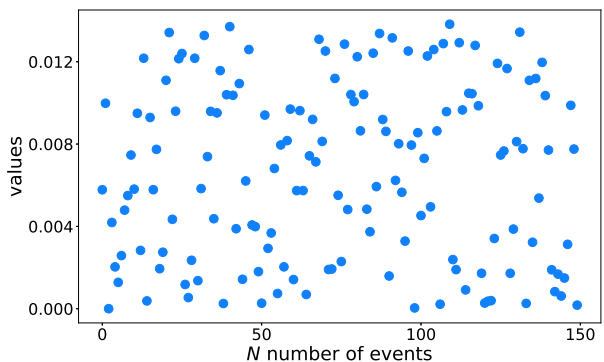}
	\caption{Random processes for $N = 150$ events.}
	\label{f:process2}
\end{figure}

To illustrate our discussion, we have generated two random processes with uncorrelated sources, that we shall use as datasets to compute the average lengths $L\pm$ as well as $L_S$,  the usual length in information theory \cite{feinstein}.

Considering the random process in Fig. \ref{f:process1}, with a binary alphabet $D=2$, we have obtained an average length  $L_S^*$=4.6bits, whereas the average lengths regarding entropies $H^{+}(P)$ and $H^{-}(P)$ are $L_{+}^*$=4.1bits and $L_{-}^*$=4.75bits, respectively. That means that a more efficient transmission process would result from the entropy $H^{+}(P)$ in comparison to a code compressed via $H^S(P)$. However as the number of random events increases, the lengths $L_{\pm}^*$ tend to coincide with $L_S^*$. For instance, with respect to the process in Fig. \ref{f:process2}, the standard average length is $L_S^{*}$=7.42bits, while $L_{+}^*$=7.3bits and $L_{-}^*$=7.49bits, hence diminishing the difference between $L_S^*$ and $L_{\pm}^*$. This is a natural consequence since the entropies $H_D^{\pm}(P)$ converge asymptotically to $H_D^S(P)$, see Ref. \cite{paper}, but maintaining a subtle difference between $H_D^{\pm}(P)$ and $H_D^S(P)$ whenever the number of events is reasonably small.

%	------------------------------------------------------------/
%
%
%	this is a blank space
%
%
%	------------------------------------------------------------/
\section{Channel capacity}
\label{s:capacity}

It would be apparently simple to expand the capacity of a communication system by merely increasing the number of different  signalling events transiting from $A$ to $B$. Why not sending  indiscriminately thousands of different voltages per symbol across the communication channel? As the number of different signals grows at $A$, a special difficulty arises at $B$. All these emitted signals are limited by the environmental noise that cannot be entirely avoided. Thus the difference between signalling events shall be greater than the intrinsic noise level to keep ambiguities away from the recovered signals by the acts of $B$.
 
To circumvent this problem, the notion of {\it channel capacity} becomes fundamental to know the maximum amount of symbols per unit time, $n$, that can be emitted and the many differences per symbols, $s$, that can be selected from the message space $s^n$.

This leads to a decision tree on the actual possibilities that spans the message space, therefore, producing a probability distribution on the codewords computed at $B$. Yet, a one-to-one relation between the input and the output is not given a priori, i.e for two different inputs we can obtain the same outcome, which makes the signals to look ambiguous. A possible choice of unambiguous inputs is realisable with a BSC, permitting a signal reconstruction with an error of no importance.

We shall use such scheme in the next few paragraphs to determine the maximum rate at which the signal can be reconstructed in terms of the entropies \eqref{entropies}, leading to a new set of bounds on the channel capacity due to Shannon.

\begin{definition}
Let $P$ and $Q$ be two different probability distributions. The amount of information shared between them is measured via the mutual information $I(P,Q)=H_D(Q) - H_D(Q\vert P)$, where $H_D(Q\vert P)$ is the conditional entropy. Then the channel capacity is defined as:
\begin{equation}
\label{capacity}
C=\max_P I(P,Q).
\end{equation}
\end{definition}

The BSC, see the diagram below, is one the simplest models of communication channels with errors. Each input is complemented with probability $p$.

%	======================================================= #
\tikzset{every picture/.style={line width=0.75pt}}
\begin{tikzpicture}[x=0.8pt,y=0.75pt,yscale=-1,xscale=1]
%Straight Lines [id:da5431747633710649] 
\draw (231.2,103.6) -- (470.7,103.6) ;
\draw [shift={(472.7,103.6)}, rotate = 180] [color={rgb, 255:red, 0; green, 0; blue, 0 }  ][line width=0.75]    (10.93,-4.9) .. controls (6.95,-2.3) and (3.31,-0.67) .. (0,0) .. controls (3.31,0.67) and (6.95,2.3) .. (10.93,4.9)   ;
%Straight Lines [id:da7136231696096738] 
\draw    (231.2,193.6) -- (470.7,193.6) ;
\draw [shift={(472.7,193.6)}, rotate = 180] [color={rgb, 255:red, 0; green, 0; blue, 0 }  ][line width=0.75]    (10.93,-4.9) .. controls (6.95,-2.3) and (3.31,-0.67) .. (0,0) .. controls (3.31,0.67) and (6.95,2.3) .. (10.93,4.9)   ;
%Straight Lines [id:da3324469565462067] 
\draw    (231.2,103.6) -- (470.83,192.9) ;
\draw [shift={(472.7,193.6)}, rotate = 200.44] [color={rgb, 255:red, 0; green, 0; blue, 0 }  ][line width=0.75]    (10.93,-4.9) .. controls (6.95,-2.3) and (3.31,-0.67) .. (0,0) .. controls (3.31,0.67) and (6.95,2.3) .. (10.93,4.9)   ;
%Straight Lines [id:da21383457518593885] 
\draw    (231.2,193.6) -- (470.83,104.3) ;
\draw [shift={(472.7,103.6)}, rotate = 519.56] [color={rgb, 255:red, 0; green, 0; blue, 0 }  ][line width=0.75]    (10.93,-4.9) .. controls (6.95,-2.3) and (3.31,-0.67) .. (0,0) .. controls (3.31,0.67) and (6.95,2.3) .. (10.93,4.9)   ;

% Text Node
\draw (265,125) node  [font=\footnotesize]  {$p$};
% Text Node
\draw (265,167) node  [font=\footnotesize]  {$p$};
% Text Node
\draw (348,93.6) node  [font=\footnotesize]  {$1-p$};
% Text Node
\draw (348,202.6) node  [font=\footnotesize]  {$1-p$};
% Text Node
\draw (221,192.6) node  [font=\footnotesize]  {$1$};
% Text Node
\draw (487,192.6) node  [font=\footnotesize]  {$1$};
% Text Node
\draw (221,103.6) node  [font=\footnotesize]  {$0$};
% Text Node
\draw (487,103.6) node  [font=\footnotesize]  {$0$};
\end{tikzpicture}
%	======================================================= #

In case that an error occurs, an input value 1 (0) will be regarded as 0 (1), hence we are unable to identify those error bits and at the end we get an untrusted set of messages. For that reason, we shall assume that every bit sent by $A$ has a negligible probability of error.

To calculate the channel capacity of a BSC, first note the mutual information $I(P,Q)$ is bounded by \cite{cover}
\begin{equation}
I(P,Q)\leq 1 - H(P),
\end{equation}
(henceforth all operations are binary, $D=2$, thus we preferably  drop such label from our notation), observe also, that the equality is attained when a uniform distribution feeds the input, it follows from \eqref{capacity} that 
\begin{equation}
\label{binary}
C_\text{BSC}=1-H(P),
\end{equation}
and we have the general expression for the channel capacity of a BSC.

%\begin{figure}[!h]
%	\centering
%	\includegraphics[width=0.5\textwidth]{capacity}
%	\caption{Capacity of a symmetric channel based on different entropy measures, with a fixed probability $p=\frac{1}{4}$ but varying the free parameter $q$ in the case of Tsallis and R\'enyi entropies.}
%	\label{f:cap1}
%\end{figure}

\begin{figure}[!h]	
	\centering
	\includegraphics[width=0.5\textwidth]{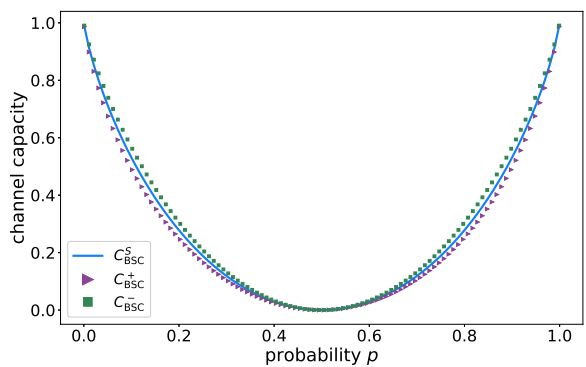}
	\caption{Capacity of a BSC in terms of entropy measures that depend only on the probability distribution. The entropy $H^-(P)$ estimates a slightly greater capacity than the one obtained from the Shannon's entropy.}
	\label{f:cap}
\end{figure}

Particularly, the channel capacity \eqref{binary} of a BSC communication system whose information weight is the Shannon's entropy, reads
\begin{equation}
\label{cs}
C_\text{BSC}^S=1+ p\log p+(1-p)\log(1-p),
\end{equation}
where the logarithms are base 2.

Likewise, we can determine the channel capacity of a BSC in terms of the entropies $H_D^\pm(P)$, which yields
\begin{equation}
\label{cp}
C_\text{BSC}^+=\frac{\sqrt{2} - p^p - (1-p)^{1-p}}{\sqrt{2}-2},
\end{equation}
and
\begin{equation}
\label{cm}
C_\text{BSC}^-=\frac{2\sqrt{2} p^{-p} - (1-p)^{-(1-p)}}{2\sqrt{2}-2},
\end{equation}
where we have normalised the capacities $C_\text{BSC}^\pm$ to be comparable with the standard case $C_\text{BSC}^S$, see Fig. \ref{f:cap}. 

In particular, the channel capacity $C_\text{BSC}^-$ exhibits a modest improvement with respect to $C_\text{BSC}^S$, indicating a tentative upper limit, above Shannon's, on how many bits can be transmitted per second over the channel without errors. Consistently, at $p=\frac{1}{2}$ the three capacities coincide, since at this point occurs the highest degree of uncertainty, resembling the initial situation in which one cannot form any judgment from the received bits at $B$.

Another interesting scenario that we shall discuss is that of a BEC: A situation in which a fraction $\alpha$ of bits is erased or lost during the transmission process, but the receiver actually knows those bits. As in the case of BSC, there are two inputs but now there will be three outputs:

%	======================================================= #

\tikzset{every picture/.style={line width=0.75pt}}       
\begin{tikzpicture}[x=0.75pt,y=0.75pt,yscale=-1,xscale=1]
%Straight Lines [id:da9449747359945333] 
\draw    (224,100) -- (474,99.5) ;
\draw [shift={(476,99.5)}, rotate = 539.89] [color={rgb, 255:red, 0; green, 0; blue, 0 }  ][line width=0.75]    (10.93,-4.9) .. controls (6.95,-2.3) and (3.31,-0.67) .. (0,0) .. controls (3.31,0.67) and (6.95,2.3) .. (10.93,4.9)   ;
%Straight Lines [id:da6610356695954198] 
\draw    (224,201) -- (474,200.5) ;
\draw [shift={(476,200.5)}, rotate = 539.89] [color={rgb, 255:red, 0; green, 0; blue, 0 }  ][line width=0.75]    (10.93,-4.9) .. controls (6.95,-2.3) and (3.31,-0.67) .. (0,0) .. controls (3.31,0.67) and (6.95,2.3) .. (10.93,4.9)   ;
%Straight Lines [id:da5513012844964188] 
\draw    (224,100) -- (472.04,149.11) ;
\draw [shift={(474,149.5)}, rotate = 191.2] [color={rgb, 255:red, 0; green, 0; blue, 0 }  ][line width=0.75]    (10.93,-4.9) .. controls (6.95,-2.3) and (3.31,-0.67) .. (0,0) .. controls (3.31,0.67) and (6.95,2.3) .. (10.93,4.9)   ;
%Straight Lines [id:da9920446710115874] 
\draw    (224,201) -- (472.04,149.9) ;
\draw [shift={(474,149.5)}, rotate = 528.36] [color={rgb, 255:red, 0; green, 0; blue, 0 }  ][line width=0.75]    (10.93,-4.9) .. controls (6.95,-2.3) and (3.31,-0.67) .. (0,0) .. controls (3.31,0.67) and (6.95,2.3) .. (10.93,4.9)   ;
% Text Node
\draw (210,101) node  [font=\footnotesize]  {$0$};
% Text Node
\draw (484,93.5) node [anchor=north west][inner sep=0.75pt]  [font=\footnotesize]  {$0$};
% Text Node
\draw (210,201) node  [font=\footnotesize]  {$1$};
% Text Node
\draw (484,193.5) node [anchor=north west][inner sep=0.75pt]  [font=\footnotesize]  {$1$};
% Text Node
\draw (490,149.5) node  [font=\footnotesize]  {$*$};
% Text Node
\draw (361.5,114) node [anchor=north west][inner sep=0.75pt]  [font=\footnotesize]  {$\alpha $};
% Text Node
\draw (361.5,175) node [anchor=north west][inner sep=0.75pt]  [font=\footnotesize]  {$\alpha $};
% Text Node
\draw (343,85.5) node [anchor=north west][inner sep=0.75pt]  [font=\footnotesize]  {$1-\alpha $};
% Text Node
\draw (343,205) node [anchor=north west][inner sep=0.75pt]  [font=\footnotesize]  {$1-\alpha $};

\end{tikzpicture}
%	======================================================= #

This diagram means that there is an input $X$, which emits 0s and 1s with probabilities $p$ and $1-p$, respectively. Whereas the output casts 0s, 1s or $*$s with probabilities $(1-\alpha)p$, $(1-\alpha)(1-p)$ or $\alpha$, in that order. Furthermore, if the source emits a 0 or 1, the probabilities of receiving these bits without errors are $(1-\alpha)p$ and $(1-\alpha)(1-p)$, otherwise the probabilities that these bits become interchanged over the communication channel are $p\alpha$ and $(1-p)\alpha$.

It follows that the channel capacity of a BEC with an erasure probability $\alpha$, is calculated as:
\begin{equation}
\label{error}
C_\text{BEC}=\max_P\, H(X) + H(Y) - H(X,Y),
\end{equation}
where $H(X,Y)$ is the joint entropy computed in terms of the joint distributions, which alternatively can be done via transition matrices, see for instance \cite{ilic}.

In the case of the Shannon's entropy, the calculation of the  capacity \eqref{error} gives:
\begin{equation}
\label{erasure}
C_\text{BEC}^S=\max_P\,(1-\alpha)H^S(P), \quad P=\{p,1-p\},
\end{equation}
since the entropy attains its maximum for a uniform distribution, $p=\frac{1}{2}$, the channel capacity of BEC reduces to $C^S=1-\alpha$, a well known result in the Shannon's theory.

\begin{figure}[!h]	
	\centering
	\includegraphics[width=0.5\textwidth]{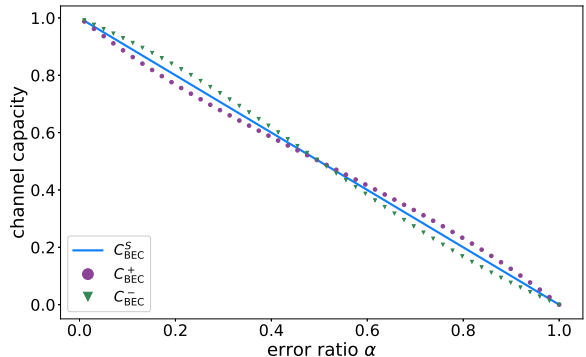}
	\caption{Channel capacity of a BEC, for entropy measures that depend only on the distribution $P$.}
	\label{f:erasure}
\end{figure}

It shall not escape attention, that the channel capacities \eqref{cs} and \eqref{erasure} are  completely additive, that is, given two distributions $p_1$ and $p_2$, we have $C^S(p_1p_2)=C^S(p_1)+C^S(p_2)$, a direct result that follows from $C^S(p_1p_2)=\text{max}_{p_{X_1,X_2}} I(X_1,X_2;Y_1,Y_2)$ with $I(X_1,X_2;Y_1,Y_2)=I(X_1,Y_1)+I(X_2,Y_2)$. Although this property will not be satisfied by the non-extensive entropies $H^\pm(P)$, in which  case the additivity is asymptotically achieved, rather we are not interested in such aspects but precisely in the non-additive consequences and possible applications, if any.

In turn, we shall compute the channel capacity of a BEC but in terms of \eqref{entropies}. As in the example of a BSC, we have to renormalise the corresponding capacities to compare with the  standard result \eqref{error}. We get:
\begin{equation}
C_\text{BEC}^+=\max_P\,\frac{1-\alpha^\alpha+H^+(P) - H^+(P\alpha)}{2-\sqrt{2}},
\end{equation}
and
\begin{equation}
C_\text{BEC}^-=\max_P\,\frac{\alpha^{-\alpha}-1+H^-(P)-H^-(P\alpha)}{2\sqrt{2}-2},
\end{equation}
as in the case \eqref{erasure}, the maximum will be attained at $p=\frac{1}{2}$. 

See Fig. \ref{f:erasure}, where we have plotted the channel capacities $C_\text{BEC}^S$ and $C_\text{BEC}^\pm$.  Interestingly $C_\text{BEC}^\pm$ behave such that the three measures coincide at the critical points $p=0,p=1$ and $p=\frac{1}{2}$, but $C_\text{BEC}^\pm$ flips their character at $p=\frac{1}{2}$, which shows the flexibility of our approach.

The flipping behaviour of $C_\text{BEC}^\pm$ is interpreted as follows. Note that $C_\text{BEC}^-$ will establish an upper bound on $C_\text{BEC}^S$ as long as the fraction of errors $\alpha$ does not dominate the communication channel, i.e. $<50\%$. Otherwise there is a tradeoff and $C_\text{BEC}^-$ will be slightly reduced immediately after the threshold at $p=\frac 1 2$. For the same reason, the lower bound suggested by $C_\text{BEC}^+$, will eventually be promoted to an upper bound, with respect to $C_\text{BEC}^S$, when the ratio $\alpha$ dominates the channel ($>50\%$). 

%for instance, consider the expansion:
%\begin{equation}
%\begin{split}
%C_\text{BEC}^+\approx&(1-\alpha)\\
%&-\frac{1}{2}\alpha\log\alpha+\frac{3}{8}\left(1-\frac{2}{3}\log\alpha-\log^2\alpha\right)\alpha^2
%\end{split}
%\end{equation}

%	------------------------------------------------------------/
%
%
%	this is a blank space
%
%
%	------------------------------------------------------------/
\section{Conclusions and Outlook}
\label{s:conclusions}

In the first part of this paper we have utilised a variational method to minimise the average length of a given prefix-free code on the basis of the non-extensive entropy measures $H_D^\pm(P)$, which depend only on the distribution. We have shown that such optimisation problem is solved as long as the adequate constraint is introduced ---in our case $K_\pm$.

These constraints defined in Prop. \ref{p1} are not arbitrary, provided the average lengths $L_\pm$ must equal the entropies $H_D^\pm(P)$ when they reach their optimal values.

In consequence we have introduced the Theorems \ref{t:h1} and \ref{t:h2}, that account for an effective data compression at different average rates, either in terms of $H_D^+(P)$ or $H_D^-(P)$. Therefore the corresponding average lengths $L_{\pm}$ must be bounded from below and above in terms of the information measures themselves and the associated optimal individual lengths $l_{\pm}(x_i)$.

Specifically, we think that $H_D^{+}(P)$ could bring some novelty with respect to the data compression estimated via the Shannon's entropy $H_D^S(P)$, iff the involved system possess only few accessible states. We suggest that such efficiency, however, could be prolonged for even a larger system if one pinpoints a reasonable way to divide it into a number of subsystems in order to outperform successively the data compression.

In the second part, we have analysed the channel capacity of two simple and generic models, a BSC and a BEC, regarding the entropies $H_D^\pm(P)$. We have obtained a channel capacity of a BSC higher than the calculated using the Shannon's theory, $C_\text{BSC}^-\geq C_\text{BSC}^S$, where the equality is attained at the critical points $p=0, p=1$ and $p=\frac{1}{2}$.

Yet an interesting behaviour occurs in the capacities of a BEC. We have obtained a twofold behaviour in each of the capacities $C_\text{BEC}^+$ and $C_\text{BEC}^-$. The first one estimates a lower bound on $C_\text{BEC}^S$ but only if the population of errors is below the $50\%$ of the received bits, otherwise it will establish a new bound, but now above $C_\text{BEC}^S$, whereas the inverse situation is furnished by $C_\text{BEC}^-$. We recall that such capacities are not strictly computed as functions of some distribution $P$, but depend now on a new parameter $\alpha$ that represents the ratio of errors over the communication channel.

%	------------------------------------------------------------/
%
%
%	this is a blank space
%
%
%	------------------------------------------------------------/
\section*{Acknowledgments}
The authors appreciate the early discussions and observations given by H. Garc\'ia-Compe\'an. We also are grateful for the general comments provided by J.L. L\'opez. J.F. would like to thank the financial support granted by CONACYT (Mexico). O.O. thanks the support of CONACYT Project 257919, UG Projects and PRODEP.
%	------------------------------------------------------------/
%
%
%	this is a blank space
%
%
%	------------------------------------------------------------/
\section*{Data Availability}
The data that support the findings of this study are available from the corresponding author upon reasonable request.
%	------------------------------------------------------------/
%
%
%	this is a blank space
%
%
%	------------------------------------------------------------/
\appendix
\section{Generalised logarithms and exponentials}
\label{app}

The generalised logarithm functions define as
\begin{equation}
\label{logs}
\begin{split}
\log^{+}(x) & \equiv-\frac{1-x^x}{x}\\
\log^{-}(x) & \equiv -\frac{x^{-x}-1}{x},
\end{split}
\end{equation}
for $x>0$, otherwise the functions become undefined. From such definitions it becomes evident that the functions $\log^{(\pm)}$ do not fulfil the three laws of logarithms. In addition, the corresponding inverse functions of the generalised logarithms \eqref{logs} do not posses a closed form hence, to subdue this technicality, a numerical representation has to be taken into account. These functions have been constructed as
\begin{equation}
\label{exps}
\exp^{\pm}(x)\equiv\exp(-x)\sum_{j=0}^\infty a_{\pm}(j) x^{j}, \quad a_\pm(j) \in \mathbb{R},
\end{equation}
whose first nine coefficients $a_\pm(j)$ are given in Table \ref{t:apm}.

\vspace{0.25cm}
\renewcommand{\arraystretch}{1.4}
\begin{table}[!htp]
\centering
\begin{tabular}{r>{\raggedleft}p{0.25\linewidth}>{\raggedleft\arraybackslash}p{0.25\linewidth}}
\hline
& $a_{+}(j)$ & $a_{-}(j)$ \\ \hline
$j=8$ & -0.000157095 & 0.000105402 \\
$j=7$ & 0.00373467   & -0.00211934 \\
$j=6$ & -0.0362676   & 0.0166679 \\
$j=5$ & 0.186358     & -0.0675544 \\
$j=4$ & -0.546751    & 0.16867 \\
$j=3$ & 0.905157     & -0.317048 \\
$j=2$ & -0.709322    & 0.3725 \\
$j=1$ & 0.0228963    & 0.0147449 \\
$j=0$ & 1 & 1 \\ \hline
\end{tabular}
\caption{$a^{(\pm)}$ Coefficients.}
\label{t:apm}
\end{table}
%	------------------------------------------------------------/
%
%
%	this is a blank space
%
%
%	------------------------------------------------------------/
\newpage
\bibliography{bibliography}

\begin{thebibliography}{10}

\bibitem{shannon}
C.~E. {Shannon}, ``A mathematical theory of communication,'' {\em The Bell
  System Technical Journal}, vol.~27, pp.~379--423, July 1948.

\bibitem{crisanti}
A.~Crisanti, A.~Puglisi, and D.~Villamaina, ``Nonequilibrium and information:
  The role of cross correlations,'' {\em Phys. Rev. E}, vol.~85, p.~061127, Jun
  2012.

\bibitem{paper}
N.~C. Bizet, J.~Fuentes, and O.~Obreg{\'{o}}n, ``Generalised asymptotic classes
  for additive and non-additive entropies,'' {\em {EPL} (Europhysics Letters)},
  vol.~128, p.~60004, feb 2020.

\bibitem{oo10}
O.~Obreg\'on, ``Superstatistics and gravitation,'' {\em Entropy}, vol.~12,
  p.~2067, 2010.

\bibitem{beck}
C.~Beck and E.~Cohen, ``Superstatistics,'' {\em Physica A: Statistical
  Mechanics and its Applications}, vol.~322, pp.~267 -- 275, 2003.

\bibitem{oo18}
O.~Obreg\'on, ``Generalized information and entanglement entropy, gravitation
  and holography,'' {\em Int. J. Mod. Phys. A}, vol.~30, no.~16, p.~1530039,
  2015.

\bibitem{tsallis}
C.~Tsallis, ``Possible generalization of {B}oltzmann-{G}ibbs statistics,'' {\em
  J. Stat. Phys.}, vol.~52, p.~479, 1988.

\bibitem{havrda}
J.~Havrda and F.~Charv\'at, ``Quantification method of classification
  processes,'' {\em Kybernetika}, vol.~3, pp.~30--35, 1967.

\bibitem{physica}
G.~Kaniadakis, M.~Lissia, and A.~Rapisarda, ``Non extensive thermodynamics and
  its applications,'' {\em Physica A: Stat. Mechs. Appl.}, vol.~305, no.~1-2,
  pp.~1--305, 2002.

\bibitem{chapeau}
F.~{Chapeau-Blondeau}, A.~{Delahaies}, and D.~{Rousseau}, ``Source coding with
  {T}sallis entropy,'' {\em Electronics Letters}, vol.~47, pp.~187--188,
  February 2011.

\bibitem{abe10}
S.~Abe, ``Essential discreteness in generalized thermostatistics with
  non-logarithmic entropy,'' {\em {EPL} (Europhysics Letters)}, vol.~90,
  p.~50004, jun 2010.

\bibitem{oikonomou}
T.~Oikonomou and G.~B. Bagci, ``Route from discreteness to the continuum for
  the tsallis $q$-entropy,'' {\em Phys. Rev. E}, vol.~97, p.~012104, Jan 2018.

\bibitem{campbell0}
L.~L. Campbell, ``Definition of entropy by means of a coding problem,'' {\em
  Zeitschrift f\"ur Wahrscheinlichkeitstheorie und Verwandte Gebiete}, vol.~6,
  pp.~113--118, Jun 1966.

\bibitem{baer}
M.~B. {Baer}, ``Source coding for quasiarithmetic penalties,'' {\em IEEE
  Transactions on Information Theory}, vol.~52, pp.~4380--4393, Oct 2006.

\bibitem{cover}
T.~M. Cover and J.~Thomas, {\em Elements of Information Theory}.
\newblock Wiley-Interscience, 2nd~ed., 2012.

\bibitem{feinstein}
A.~Feinstein, {\em Foundations of information theory}.
\newblock McGraw-Hill electrical and electronic engineering series, McGraw
  Hill: New York, 1958.

\bibitem{ilic}
V.~M. Ilic, I.~Djordjevic, and F.~K{\"u}ppers, ``On the dar{\'o}czy-tsallis
  capacities of discrete channels,'' 2015.

\end{thebibliography}
%	------------------------------------------------------------/
\end{document}